\documentclass[11pt]{article}

\usepackage[T1]{fontenc}
\usepackage{newpxtext}
\usepackage[protrusion=true,expansion=true]{microtype}
\usepackage[margin=1.12in,headheight=14pt]{geometry}
\usepackage[dvipsnames]{xcolor}
\definecolor{ink}{HTML}{253447}
\definecolor{accent}{HTML}{2F6F73}
\definecolor{muted}{HTML}{687681}
\usepackage{amsmath,amsthm,amsfonts,amssymb}
\usepackage[numbers,sort]{natbib}
\usepackage{graphicx}
\usepackage{fancyhdr}
\usepackage{epstopdf}
\usepackage[all,cmtip]{xy}
\usepackage{stmaryrd}
\usepackage{stackengine}
\usepackage{mathtools}
\usepackage{subcaption}
\usepackage{tikz}
\usepackage[textsize=footnotesize]{todonotes}
\usepackage{mathrsfs}
\usepackage[toctitles]{titlesec}
\usepackage{enumitem}
\usepackage{titling}
\usepackage[pagebackref,colorlinks,
    linkcolor=accent,
    citecolor=accent,
    urlcolor=accent]{hyperref}
\usepackage{doi}

\usepackage{aliascnt}
\usepackage[capitalize,noabbrev]{cleveref}

\titleformat{\section}
  {\Large\bfseries\color{ink}}
  {\thesection}{0.65em}{}
  [\vspace{0.25em}{\color{accent!55}\titlerule[0.6pt]}]
\titlespacing*{\section}{0pt}{2.4em}{1.15em}
\setlist[itemize]{leftmargin=1.5em,itemsep=0.25em,topsep=0.45em}

\fancypagestyle{plain}{\fancyhf{}}

\pretitle{\begin{center}\color{ink}\LARGE\bfseries}
\posttitle{\par\vspace{0.7em}{\color{accent}\rule{0.18\textwidth}{1.2pt}}\end{center}\vspace{0.2em}}
\preauthor{\begin{center}\small\color{muted}}
\postauthor{\end{center}}
\predate{\begin{center}\small\color{muted}}
\postdate{\end{center}\vspace{1.25em}}

\newcommand\barbelow[1]{\stackunder[1.5pt]{$\hspace{2pt}#1\hspace{2pt}$}{\rule{1ex}{0.075ex}}}

\newcommand{\barabove}[1]{\hspace{2pt}\bar{#1}\hspace{2pt}}

\newtheoremstyle{result}
  {0.9em}{0.9em}
  {\itshape}{}
  {\color{ink}\bfseries}{.}{0.55em}{}
\newtheoremstyle{discussion}
  {0.9em}{0.9em}
  {}{}
  {\color{accent}\bfseries}{.}{0.55em}{}

\theoremstyle{result}
\newtheorem{theorem}{Theorem}[section]
\newaliascnt{lemma}{theorem}
\newtheorem{lemma}[lemma]{Lemma}
\aliascntresetthe{lemma}
\newaliascnt{lem}{theorem}

\aliascntresetthe{lem}
\newaliascnt{corollary}{theorem}
\newtheorem{corollary}[corollary]{Corollary}
\aliascntresetthe{corollary}
\newaliascnt{proposition}{theorem}
\newtheorem{proposition}[proposition]{Proposition}
\aliascntresetthe{proposition}
\theoremstyle{discussion}
\newaliascnt{question}{theorem}

\aliascntresetthe{question}
\newaliascnt{conjecture}{theorem}

\aliascntresetthe{conjecture}
\newaliascnt{remark}{theorem}
\newtheorem{remark}[remark]{Remark}
\aliascntresetthe{remark}
\newaliascnt{definition}{theorem}
\newtheorem{definition}[definition]{Definition}
\aliascntresetthe{definition}
\newaliascnt{example}{theorem}
\newtheorem{example}[example]{Example}
\aliascntresetthe{example}
\newaliascnt{problem}{theorem}

\aliascntresetthe{problem}

\crefname{proposition}{proposition}{propositions}
\Crefname{proposition}{Proposition}{Propositions}

\AtBeginEnvironment{example}{%
  \pushQED{\qed}%
}
\AtEndEnvironment{example}{\popQED}
\AtBeginEnvironment{remark}{%
  \pushQED{\qed}%
}
\AtEndEnvironment{remark}{\popQED}

\AtBeginEnvironment{definition}{%
  \pushQED{\qed}%
}
\AtEndEnvironment{definition}{\popQED}

\makeatletter
\newcommand{\suppressnextindent}{\@afterindentfalse\@afterheading}
\makeatother
\AfterEndEnvironment{theorem}{\suppressnextindent}
\AfterEndEnvironment{lemma}{\suppressnextindent}
\AfterEndEnvironment{lem}{\suppressnextindent}
\AfterEndEnvironment{corollary}{\suppressnextindent}
\AfterEndEnvironment{proposition}{\suppressnextindent}
\AfterEndEnvironment{question}{\suppressnextindent}
\AfterEndEnvironment{conjecture}{\suppressnextindent}
\AfterEndEnvironment{remark}{\suppressnextindent}
\AfterEndEnvironment{definition}{\suppressnextindent}
\AfterEndEnvironment{example}{\suppressnextindent}
\AfterEndEnvironment{problem}{\suppressnextindent}

\usepackage{makecell}

\usepackage{arydshln}
\usepackage{authblk}
\newcommand{\omin}{\barbelow{\otimes}}
\newcommand{\omax}{\barabove{\otimes}}

\title{Purifications for Convex Cones}
\author{Felix Campidell}
\author{Tim Netzer}
\affil{Department of Mathematics, University of Innsbruck, Austria}
\date{\today}

\begin{document}
\maketitle

\begin{abstract}
Motivated by the importance of the purification principle in quantum theory and generalized probabilistic
theories, we study purifications using only the geometry of a
finite-dimensional proper convex cone.
We prove an existence theorem for indecomposable cones and intermediate
tensor cones containing the maximally entangled state; in particular, every
interior point of an indecomposable homogeneous cone admits a purification. This applies to Lorentz cones, for example.
We also give a criterion for uniqueness up to local automorphisms.  On the
boundary, we show that if every proper face of $C$ is simplicial, then only
pure points can admit purifications, and we demonstrate that this conclusion
fails in the presence of non-simplicial faces.  Examples involving positive
semidefinite cones, Lorentz cones, $k$-positive maps, PPT tensors, and
polyhedral cones illustrate both existence and non-uniqueness phenomena.
\end{abstract}

\section{Introduction}

Generalized probabilistic theories provide an operational framework in which
classical and quantum theory, as well as more general hypothetical theories,
can be studied in a common language \citep{Barrett2007,Plavala2023}.  In the
finite-dimensional setting, a physical system is described by a proper
convex cone of unnormalized states, its dual cone describes effects, and
physical transformations are positive linear maps.  The possible composites
of two systems are encoded by tensor cones lying between the minimal and
maximal tensor products.  This formulation separates the probabilistic and
convex-geometric structure of a theory from the particular Hilbert-space
representation used in quantum mechanics.

One of the most powerful principles considered in this framework is
\emph{purification}: every mixed state should arise as the marginal of a pure
state of a larger system, and purifications should be unique up to reversible
transformations on the purifying system.  This principle is not merely a way
of representing mixed states. For instance, it entails many characteristic features of quantum
information theory, including perfect correlations, no information without
disturbance, no cloning, teleportation, and quantum error-correction
phenomena \citep{ChiribellaEtAl2010}.  Together with a small collection of
standard operational axioms, purification singles out finite-dimensional
quantum theory \citep{ChiribellaEtAl2011}.  It is therefore natural to ask
which part of the purification principle follows from the convex geometry of
a state cone alone.

We study this question for an arbitrary finite-dimensional proper cone
$C$.  After fixing an interior point $u\in C$, we call an extreme element of
$C^\vee\omax C$ a purification of $c\in C$ if its marginal, obtained by evaluating the first tensor factor in $u$, is $c$.  More
generally, we consider purifications lying in an intermediate cone
\[
    P\subseteq C^\vee\omax C.
\]
Under the canonical identification
$C^\vee\omax C={\rm Pos}(C)$, a purification is precisely an extreme
positive map $h$ satisfying $h(u)=c$.  Thus existence and uniqueness of
purifications become questions about the extreme rays of the cone of
positive endomorphisms of $C$.

Our main existence result shows that, when $C$ is indecomposable and $P$
contains the maximally entangled state, every point in a natural
automorphism orbit of $u$ has a purification in $P$.  In particular, every
interior point of an indecomposable homogeneous cone has a purification in
the maximal tensor product. On the
boundary, we prove that if every proper face of $C$ is simplicial, then every
nonzero boundary point admitting a purification is pure, and we give an
example showing that non-pure boundary points may be purified when a
non-simplicial face is present.

We next give a criterion for uniqueness in an intermediate cone $P$: it is
controlled by the extreme positive maps fixing $u$ and by the local
automorphisms preserving $P$.  The final section illustrates the results for
positive semidefinite cones, Lorentz cones, $k$-positive maps, the PPT cone,
and polyhedral cones.  These examples exhibit both uniqueness and
non-uniqueness, as well as the dependence of purification on the chosen
composite cone.

\section{Preliminaries}

Throughout this paper, all vector spaces are real and finite-dimensional.
Unless stated otherwise, $C\subseteq V$ and $D\subseteq W$ denote
\emph{proper} convex cones: they are closed, full-dimensional, and salient,
that is, $C\cap(-C)=\{0\}$. We write ${\rm int}(C)$ for the interior of
$C$ and ${\rm ex}(C)$ for the set of nonzero elements spanning extreme rays
of $C$. Elements of ${\rm ex}(C)$ are also called \emph{pure}.

The cone of positive linear maps from $C$ to $D$ is
\[
    {\rm Pos}(C,D)
    =\{g\in{\rm Lin}(V,W)\mid g(C)\subseteq D\},
\]
and we abbreviate ${\rm Pos}(C,C)$ to ${\rm Pos}(C)$. The automorphism group
of $C$ is
\[
    {\rm Aut}(C)
    =\{g\in{\rm GL}(V)\mid g(C)=C\}.
\]

The dual cone of $C$ is
\[
    C^\vee=\{f\in V'\mid f(c)\geqslant0\text{ for all }c\in C\}.
\]
It is again proper. A \emph{unit effect} on $C$ is a functional
$\varphi\in C^\vee$ that is strictly positive on $C\setminus\{0\}$.

For convex cones $C,D$ we use the minimal tensor product 
$$
    C\omin D=\left\{ \sum_i c_i\otimes d_i\mid c_i\in C,d_i\in D\right\}
$$
and the maximal tensor product
$$
    C\omax D=\left( C^\vee\omin D^\vee\right)^\vee.
$$ 
It is well-known and easy to see that under the standard identification 
$V'\otimes W\cong {\rm Lin}(V,W)$
we have
$$
    C^\vee\omax D = {\rm Pos}(C,D).
$$

Now let $C, D$ be convex cones with unit effects $\varphi_1,\varphi_2$. For $\omega\in C\omax D$ we write 
$$
    \omega_C\coloneq ({\rm id}\otimes \varphi_2)(\omega) \mbox{ and } \omega_D\coloneqq(\varphi_1\otimes{\rm id})(\omega)
$$ 
for the so-called \emph{marginals} of $\omega.$ By definition of the maximal tensor product have $\omega_C\in C$ and $\omega_D\in D.$

\begin{proposition}[{\citealp[Lemma~3]{BarnumEtAl2006}}]
\label{prop:pure-marginal-product}
    Let $C, D$ be convex cones with unit effects $\varphi_1,\varphi_2$, and let  $\omega\in C\omax D.$  If  $\omega_C\in{\rm ex}(C)$ or $\omega_D\in{\rm ex}(D)$, then $\omega$ is proportional to $\omega_C\otimes \omega_D.$
\end{proposition}

\begin{proposition}[{\citealp[Theorem 3.22 and Proposition 4.48]{vanDobbenDeBruyn2020}}]
\label{prop:extreme-rays-tensor-cones}
    For convex cones $C$ and $D$ we have 
    $$
        {\rm ex}(C\omin D)=\{ c\otimes d\mid c\in{\rm ex}(C), d\in{\rm ex}(D)\}\subseteq{\rm ex}(C\omax D).
    $$
\end{proposition}

\begin{corollary}
\label{cor:minmax}
    Let $C, D$ be convex cones with unit effects $\varphi_1,\varphi_2$.
    Then for $\omega\in{\rm ex}(C\omax D)$ we have 
    $$
        \omega_C\in{\rm ex}(C)\ \Leftrightarrow\  \omega\in C\omin D.
    $$
\end{corollary}

\begin{proof}
    This follows immediately from \Cref{prop:pure-marginal-product,prop:extreme-rays-tensor-cones}.
\end{proof}

\begin{proposition}[{\citealp[Theorem 2.B.2]{Barker1981}}]
\label{prop:automorphisms-extreme}
    We have $${\rm id}\in{\rm ex}({\rm Pos}(C))\Leftrightarrow {\rm Aut}(C)\subseteq{\rm ex}({\rm Pos}(C))\Leftrightarrow C \mbox{ indecomposable}.$$
\end{proposition}

\section{Purifications}

\begin{definition}[Purifications]
    Let $C$ be a convex cone equipped with a unit effect $\varphi$ on
    $C^\vee$. A \emph{purification} of $c\in C$ is an element
    \[
        \omega\in{\rm ex}(C^\vee\omax C)
        \qquad\text{such that}\qquad
        \omega_C=c.
    \]
    If, in addition, $\omega$ belongs to a convex cone
    $P\subseteq C^\vee\omax C$, we call it a \emph{purification of $c$ in
    $P$}. Purity is always understood relative to the maximal tensor cone:
    membership in ${\rm ex}(P)$ alone is not sufficient.
\end{definition}

\begin{remark}
    Let $c\in{\rm ex}(C)$ and let $\omega$ be a purification of $c$.
    Since one marginal of $\omega$ is pure,
    \Cref{prop:pure-marginal-product} implies that $\omega$ is a product
    tensor. Consequently,
    \[
        \omega\in C^\vee\omin C.
    \]
    Thus pure elements have only trivial purifications.
\end{remark}

Under the canonical identification $V\cong V''$, every unit effect
$\varphi$ on $C^\vee$ is evaluation at a unique point
$u\in{\rm int}(C)$.
We fix this choice of $u$ and $\varphi$ throughout the remainder of the
section.

\begin{lemma}
\label{lem:marginal-evaluation}
    For $\omega\in C^\vee\omax C={\rm Pos}(C)$ we have $\omega_C=\omega(u).$     
\end{lemma}

\begin{proof}
    Under the identification $V'\otimes V\cong{\rm Lin}(V,V)$, a tensor
    $\omega=\sum_i f_i\otimes x_i$ corresponds to the map
    $v\mapsto\sum_i f_i(v)x_i$. Therefore its marginal on the second factor is
    \[
        \omega_C
        =(\varphi\otimes{\rm id})(\omega)
        =\sum_i \varphi(f_i)x_i
        =\sum_i f_i(u)x_i
        =\omega(u),
    \]
    as claimed.
\end{proof}

\begin{proposition}[Cones with simplicial faces]
\label{prop:poly}
    Let $C$ be a cone such that every proper face of $C$ is
    simplicial. If  a nonzero boundary point of $C$ admits a purification, then it is pure.
\end{proposition}

\begin{proof}
    Let $h\in{\rm ex}({\rm Pos}(C))$ be a purification and suppose that
    $c=h(u)$ belongs to a proper face $F$ of $C$. For every $x\in C$ there
    exists $\lambda>0$ such that $\lambda u-x\in C$. Hence
    \[
        \lambda c=h(x)+h(\lambda u-x)\in F,
    \]
    and the facial property gives $h(x)\in F$. Thus $h(C)\subseteq F$.

    Since $F$ is simplicial, write
    \[
        F={\rm cone}\{r_1,\ldots,r_k\}
    \]
    with linearly independent generators. Since the $r_i$ linearly span the image of the map $h$, there is a tensor representation 
    $$ 
        h=\sum_{i=1}^k f_i\otimes r_i
    $$
    with $f_i\in V'.$ For all $d\in C$ we have
    $$
    F\ni h(d)=\sum_i f_i(d)r_i,
    $$
    which implies $f_i(d)\geqslant 0$ and thus $f_i\in C^\vee,$ for all $i$. So $h\in C^\vee\omin C$ and thus $c=h(u)=h_C$ is pure, by \Cref{cor:minmax}.
\end{proof}

Let $(v_i)$ be a basis of $V$ and $(v_i')$ its dual basis. We call the
tensor
\[
    m=\sum_i v_i'\otimes v_i\in V'\otimes V
\]
the \emph{maximally entangled state}. It is independent of the chosen basis
and, under the identification
$V'\otimes V\cong{\rm Lin}(V,V)$, corresponds to the identity map on $V$. In particular, $m\in C^\vee\omax C={\rm Pos}(C)$.

\begin{theorem}[Existence of purifications]
\label{cor:hompuri}
    Let $C$ be indecomposable and let
    \[
        P\subseteq C^\vee\omax C
    \]
    be a convex cone containing the maximally entangled state $m$. Define
    \[
        G_2(P)=\bigl\{g\in{\rm Aut}(C)\bigm|
        {\rm id}\otimes g\in{\rm Aut}(P)\bigr\}.
    \]
    Then every point in the orbit $G_2(P)\cdot u$ has a purification in $P$.
\end{theorem}

\begin{proof}
    Let $c=g(u)$ for some $g\in G_2(P)$. Since $m\in P$ and
    ${\rm id}\otimes g$ is an automorphism of $P$, we have
    \[
        ({\rm id}\otimes g)(m)\in P.
    \]
    Under the identification $C^\vee\omax C={\rm Pos}(C)$, this tensor
    corresponds to the automorphism $g$. By
    \Cref{prop:automorphisms-extreme}, indecomposability of $C$ implies that
    $g$ spans an extreme ray of ${\rm Pos}(C)$. Furthermore,
    \Cref{lem:marginal-evaluation} gives
    \[
        g_C=g(u)=c.
    \]
    Thus $g$ is a purification of $c$ in $P$.
\end{proof}

\begin{corollary}[Purifications for homogeneous cones]
\label{cor:exhom}
    If $C$ is indecomposable and homogeneous, then every interior point of $C$
    has a purification in $C^\vee\omax C$.
\end{corollary}

\begin{proof}
    Apply \Cref{cor:hompuri} with $P=C^\vee\omax C$. The maximally
    entangled state $m$ belongs to $P$, and the maximal tensor product is
    preserved by all local automorphisms. Hence
    \[
        G_2(P)={\rm Aut}(C).
    \]
    Since $C$ is homogeneous, ${\rm Aut}(C)$ acts transitively on
    ${\rm int}(C)$. The conclusion therefore follows from
    \Cref{cor:hompuri}.
\end{proof}

\begin{remark}
    Let $\otimes$ be a tensor product of cones that is functorial with respect to positive maps, and
    suppose that
    \[
        m\in C^\vee\otimes C.
    \]
    Then, for every $g\in{\rm Pos}(C)$, functoriality
    gives
    \[
        g=({\rm id}\otimes g)(m)
        \in C^\vee\otimes C.
    \]
    Since ${\rm Pos}(C)=C^\vee\omax C$, it follows that
    \[
        C^\vee\otimes C=C^\vee\omax C.
    \]
    Thus if $P$ from \Cref{cor:hompuri} arises from a functorial tensor product, we necessarily have $P=C^\vee\omax C$.
\end{remark}

We now turn to uniqueness of purifications.
We denote by
${\rm ex}_u({\rm Pos}(C))$ the set of pure positive maps of $C$ that fix $u$. In view of the above, this is exactly the set of purifications of $u$. For $g\in{\rm Pos}(C)$ we denote by  $g^*\in{\rm Pos}(C^\vee)$ the dual map, defined by $g^*(f)=f\circ g$.

\begin{theorem}[Uniqueness of purifications]
\label{thm:unique-purifications-subcone}
    Let $C$ be indecomposable, let
    \[
        P\subseteq C^\vee\omax C={\rm Pos}(C)
    \]
    be a convex cone containing the maximally entangled state $m$, and set
    \[
        G_1(P)=\bigl\{g\in{\rm Aut}(C)\bigm|
        g^*\otimes{\rm id}\in{\rm Aut}(P)\bigr\}.
    \]
    Then the purification $m$ of $u$ is unique among the purifications of
    $u$ that lie in $P$, up to local automorphisms of $P$ on the first
    tensor factor, if and only if
    \[
        {\rm ex}_u({\rm Pos}(C))\cap P\subseteq G_1(P).
    \]
    If these equivalent conditions hold, then every point in
    $G_2(P)\cdot u$ has a unique purification lying in $P$, up to local
    automorphisms of $P$ on the first tensor factor.
\end{theorem}

\begin{proof}
    Since $C$ is indecomposable, $m$ corresponds to the extreme positive map
    ${\rm id}$ and is therefore a purification of $u$. Under the
    identification $C^\vee\omax C={\rm Pos}(C)$, the action of
    $g^*\otimes{\rm id}$ is precomposition by $g$, that is,
    $(g^*\otimes{\rm id})(h)=h\circ g$. Hence the orbit of $m$
    under local automorphisms of $P$ on the first factor is precisely
    $G_1(P)$. The first assertion now follows because the purifications of
    $u$ lying in $P$ are exactly
    \[
        {\rm ex}_u({\rm Pos}(C))\cap P.
    \]
    Now let $c=g(u)$ with $g\in G_2(P)$. The tensor
    $({\rm id}\otimes g)(m)$ lies in $P$, corresponds to the automorphism
    $g$, and thus purifies $c$. If $h\in P$
    is any other purification of $c$, then
    $({\rm id}\otimes g^{-1})(h)$ lies in $P$, corresponds to the extreme
    positive map $g^{-1}\circ h$, and satisfies
    \[
        (g^{-1}\circ h)(u)=g^{-1}(c)=u.
    \]
    By the first assertion, $g^{-1}\circ h\in G_1(P)$. Thus $h$ is obtained from
    $g$ by a local automorphism of $P$ on the first tensor factor, proving
    the claim.
\end{proof}

\section{Applications and Examples}

\begin{example}[Purifications in quantum theory]
\label{ex:qupuri}
    Let $C={\rm Psd}_d$ be the cone of positive semidefinite complex hermitian
    $d\times d$ matrices and let $u=I_d$. Under the usual self-duality of $C$, set
    \[
        P=C^\vee\otimes_{\rm psd}C\coloneqq {\rm Psd}_{d^2}.
    \]
    The maximally entangled state $m$ is a rank-one positive semidefinite
    matrix and therefore belongs to $P$. Moreover, the congruence
    automorphisms
    \[
        g_A(X)=AXA^*,\qquad A\in{\rm GL}_d(\mathbb C),
    \]
    preserve $P$ on the second tensor factor and act transitively on
    ${\rm int}(C)$. Hence \Cref{cor:hompuri} shows that every positive-definite
    matrix has a purification in ${\rm Psd}_{d^2}$. This recovers the usual
    existence theorem for quantum purifications.

    Now let $h\in{\rm ex}_u({\rm Pos}(C))\cap P$. Since the Choi matrix of $h$ is positive
    semidefinite and $h$ is extremal, we have 
    \[
        h(X)=AXA^*
    \]
    for some matrix $A$. From  $h(I_d)=I_d$ we obtain that $A$ is a unitary. Again, since such congruences are local automorphisms of $P$ and act transitively on the interior of $C$, we obtain the well-known uniqueness of purifications for positive definite matrices from \Cref{thm:unique-purifications-subcone}.
\end{example}

\begin{theorem}[Lorentz cones]
    For $n\geqslant 3$ every interior point of the Lorentz cone $L_n\subseteq\mathbb R^n$ has a unique purification in $L_n^\vee\omax L_n$, up to local
    automorphisms on the first tensor factor.
\end{theorem}

\begin{proof}
   The Lorentz cone is
    indecomposable and homogeneous. Every extreme positive endomorphism of
    $L_n$ is either of rank one or an automorphism
    \citep[Theorem 4.5]{LoewySchneider1975}. A rank-one extreme map cannot fix an interior
    point $u$, and hence
    \[
        {\rm ex}_u({\rm Pos}(L_n))\subseteq{\rm Aut}(L_n).
    \]
    The claim thus follows from \Cref{cor:exhom} and \Cref{thm:unique-purifications-subcone}.
\end{proof}

\begin{example}
    Let $C={\rm Psd}_2$. As a real convex
    cone, $C$ is isomorphic to the Lorentz cone $L_4$. It follows that every
    positive-definite $2\times2$ matrix has a unique purification in $C^\vee\omax C,$
    up to a local automorphism on the first tensor factor.
    This is stronger than the uniqueness statement from \Cref{ex:qupuri}, which considers purifications in 
    $C^\vee\otimes_{\rm psd}C={\rm Psd}_4$ only.
\end{example}

\begin{example}[Non-uniqueness of purifications]
\label{ex:choi}
    Let $C={\rm Psd}_3$ with $u=I_3$. Besides the identity,
    consider the normalized Choi map $\Phi\colon C\to C$ given by
    \[
        \Phi(X)=\frac{1}{2}
        \begin{pmatrix}
            x_{11}+x_{33} & -x_{12} & -x_{13} \\
            -x_{21} & x_{22}+x_{11} & -x_{23} \\
            -x_{31} & -x_{32} & x_{33}+x_{22}
        \end{pmatrix}.
    \]
    This map is positive and spans an extreme ray of ${\rm Pos}(C)$
    \citep{Choi1975,Ha2013}; moreover, $\Phi(I_3)=I_3$. Thus both ${\rm id}$ and
    $\Phi$ are purifications of $I_3$ in $C^\vee\omax C$.

    These purifications are not equivalent under local automorphisms on the
    first tensor factor. Indeed, the orbit of ${\rm id}$ is ${\rm Aut}(C)$,
    whereas $\Phi$ is not an automorphism: it sends the rank-one matrix
    $E_{11}$ to $\tfrac12\operatorname{diag}(1,1,0)$, which has rank two. Since automorphisms map extreme rays (=rank one matrices) to extreme rays, $\Phi$ cannot be an automorphism.
   
    This does not contradict \Cref{ex:qupuri}, since $\Phi$ is not
    completely positive and hence does not belong to the smaller cone
    $C^\vee\otimes_{\rm psd}C={\rm Psd}_9$.
\end{example}

\begin{example}[$k$-positive maps]
    Let $C={\rm Psd}_d$ with $d\geqslant3$ and $u=I_d$. For
    $1\leqslant k\leqslant d$, let $P_k\subseteq{\rm Pos}(C)$ be the cone of
    $k$-positive maps. Thus
    \[
        P_d\subseteq\cdots\subseteq P_2\subseteq P_1={\rm Pos}(C),
    \]
    where $P_d={\rm Psd}_{d^2}$ is the cone of completely positive maps. Thus by \Cref{ex:qupuri}, every positive definite matrix has a purification in every $P_k$.

    We claim that these purifications are unique, at least for $k\geqslant2$. Indeed,
    let
    \[
        h\in{\rm ex}_u({\rm Pos}(C))\cap P_k.
    \]
    Then $h$ is an extremal positive map that is $2$-positive, and is
    therefore completely positive
    \citep[Theorem~3.4]{Marciniak2010}. Thus $h\in P_d,$ and purifications are unique here, by \Cref{ex:qupuri}. Note that congruence autormorphisms preserve all cones $P_k$, since they correspond to precompositions as maps.
    
    However, $P_1={\rm Pos}(C)$ admits
    non-unique purifications, as shown in \Cref{ex:choi} above.
\end{example}

\begin{example}[The PPT cone]
    Let $C={\rm Psd}_d$ with $d\geqslant2$, and consider the cone
    \[
        P_{\rm PPT}
        =\{c\in{\rm Psd}_{d^2}\mid \Gamma(c)\in{\rm Psd}_{d^2}\}
        \subseteq C^\vee\omax C,
    \]
    where $\Gamma$ denotes partial transpose. The maximally entangled state
    does not belong to $P_{\rm PPT}$: its partial transpose is not positive semidefinite for $d\geqslant2$.

    In fact, no positive-definite matrix has a purification in
    $P_{\rm PPT}$. To see this, let $\omega$ be a purification lying in
    $P_{\rm PPT}$. Since its Choi matrix is positive semidefinite and
    $\omega$ is extreme in ${\rm Pos}(C)$, that matrix must have rank one.
    A rank-one PPT matrix is a product matrix
    \citep[Lemma~2.14]{AubrunSzarek2017}, so its marginal has rank one,
    i.e.\ is already pure. This example shows that the assumption $m\in P$
    in \Cref{cor:hompuri} is essential.
\end{example}

\begin{example}[The square cone]
    Consider the cone $S\subseteq\mathbb R^3$, spanned by 
    $$
        (1,1,1),(1,-1,1),(-1,1,1),(-1,-1,1)
    $$
    and interior point $u=(0,0,1).$
    From \Cref{prop:poly} we know that the only boundary points with purification are the four extreme rays

    $S$  is indecomposable, and a direct but tedious computation shows that it has $8$ automorphisms up to scaling, all of which fix $u$. Further, $S^\vee\omax S$ has $24$ extreme rays, of which $16$ lie in $S^\vee\omin S$ and purify the extreme rays of $S$. The remaining $8$ correspond to automorphisms, which all purify the interior point. No other points admit purifications. All existing purifications are unique up to automorphisms.
\end{example}

\begin{example}
    Let $S\subseteq\mathbb R^3$ be the square cone from the preceding
    example and set
    \[
        C=S\oplus\mathbb R_{\geqslant0},
        \qquad
        u=(0,0,1,1)\in{\rm int}(C).
    \]
    Consider the positive map
    \[
        h\colon C\to C,
        \qquad
        h(s,t)=(s,0).
    \]
    This map is extreme. Indeed, suppose that $h=h_1+h_2$ with
    $h_i\in{\rm Pos}(C)$. Since $h(0,1)=0$, salience of $C$ gives
    $h_i(0,1)=0$. Moreover, $S\oplus\{0\}$ is a face of $C$, so the identity
    \[
        (s,0)=h_1(s,0)+h_2(s,0)
    \]
    implies that both summands lie in $S\oplus\{0\}$. Thus the restrictions
    of the $h_i$ give positive maps $k_i\in{\rm Pos}(S)$ satisfying
    ${\rm id}_S=k_1+k_2$. Since $S$ is indecomposable,
    \Cref{prop:automorphisms-extreme} implies that ${\rm id}_S$ is extreme,
    and hence each $k_i$ is a nonnegative multiple of ${\rm id}_S$.
    Therefore each $h_i$ is a nonnegative multiple of $h$.

    Consequently, $h$ is a purification of
    \[
        h(u)=(0,0,1,0).
    \]
    This point lies in the relative interior of the non-simplicial proper
    face $S\oplus\{0\}$ of $C$. In particular, it is a boundary point that
    is not pure. This shows that the simplicial-face assumption in \Cref{prop:poly} is essential.
\end{example}

\section*{Acknowledgments}
We thank Alexander M\"uller-Hermes for helpful discussions concerning
Lorentz cones.

\bibliographystyle{plainnat}
\bibliography{references}

\end{document}